\documentclass[a4paper]{article}
\usepackage{amsfonts,amssymb,amsbsy,amsmath,amsthm,enumerate}
\usepackage{color}
\usepackage{a4wide}
\DeclareMathOperator{\real}{Re}
\DeclareMathOperator{\imag}{Im}
\renewcommand{\Re}{\real}
\renewcommand{\Im}{\imag}

\newtheorem{theorem}{Theorem}[section]

\newtheorem{proposition}[theorem]{Proposition}
\newtheorem{follow}[theorem]{Corollary}
\newtheorem{lemma}[theorem]{Lemma}

\theoremstyle{definition}

\numberwithin{equation}{section}

\newtheorem{remark}[theorem]{Remark}

\newcommand{\bel}{\begin{equation} \label}
\newcommand{\ee}{\end{equation}}

\newcommand{\x}{{\bf x}}
\newcommand{\y}{{\bf y}}

\newcommand{\rt}{{\mathbb R}^{3}}

\newcommand{\re}{{\mathbb R}}

\newcommand{\A}{{\bf a}}

\newcommand{\C}{{\mathbb C}}

\def\beq{\begin{equation}}
\def\eeq{\end{equation}}
\newcommand{\bea}{\begin{align}}
\newcommand{\eea}{\end{align}}
\newcommand{\beas}{\begin{align*}}
\newcommand{\eeas}{\end{align*}}

\begin{document}
\begin{center}
{\large \bf  Diamagnetism of quantum gases with singular potentials.}

\medskip

\today
\end{center}

\begin{center}
\small{
Philippe Briet \footnote{Universit{\'e} du sud Toulon-Var \& Centre de Physique Th\'eorique,  Luminy, Case 907
13288 Marseille CEDEX 9, France; e-mail: briet@univ-tln.fr  },
Horia D. Cornean\footnote{Department of Mathematical Sciences,
    Aalborg
    University, Fredrik Bajers Vej 7G, 9220 Aalborg, Denmark; e-mail:
    cornean@math.aau.dk},
Baptiste Savoie \footnote{Centre de Physique Th\'eorique, Luminy, Case 907,
13288 Marseille,  CEDEX 9, France, and Department of Mathematical Sciences,
    Aalborg
    University, Fredrik Bajers Vej 7G, 9220 Aalborg, Denmark; e-mail:
    baptiste.savoie@gmail.com}}.

\end{center}

\vspace{0.5cm}

\begin{abstract}
We consider a gas of quasi-free quantum particles confined to a finite box,
subjected to singular magnetic and electric fields. We
prove in great generality that
the  finite volume grand-canonical pressure is
analytic  w.r.t. the chemical potential and the intensity of the
external magnetic field. We also  discuss the
thermodynamic limit.
\end{abstract}



%%%%%%%%%%%%%%%%%%%%%%%%%%%%%%%%%%%%%%%%%%%%%%%%%%%%%%%%%%%%%%%%%%%%%%%%%%
%%%%%%%%%%%%%%%%%%%%%%%%%%%%%%%%%%%%%%%%%%%%%%%%%%%%%%%%%%%%%%%%%%%%%%%%%%

%%%%%%%%%%%%%%%%%%%%%%%%%%%%%%%%%%%%%%%%%%%%%%%%%%%%%%%%%%%%%%%%%%%%%%%%%
%%%%%%%%%%%%%%%%%%%%%%%%%%%%%%%%%%%%%%%%%%%%%%%%%%%%%%%%%%%%%%%%%%%%%%%%%
\section{Introduction and the main result.}

The quest for rigorous results concerning the thermodynamic limit of
the magnetic susceptibility of a gas of quasi-free quantum particles in the presence
of a magnetic field started in 1975 with the work of Angelescu {\it
  et. al.} \cite{abn1, abn2}. Their method
 used in an
essential way that the confining domain was a parallelepiped,
the Hamiltonian was purely
magnetic and the susceptibility was computed at zero magnetic field.

In a series of papers \cite{C, BrCoLo1, BrCoLo2, BrCoLo3, CN2, Sa} we
gradually removed these constraints, and now we know how to prove the
thermodynamic limit for generalized susceptibilities at arbitrarily
large magnetic fields, and with smooth and periodic electric
potentials. This achievement was possible due to  a new idea, which
led to the development of a systematic magnetic perturbation theory
for Gibbs semigroups.

In this paper we examine the case in which both the magnetic field and
the electric potential can have singularities, such that the magnetic
and scalar singular perturbations are relatively bounded in the form
sense with respect to
the purely magnetic Schr\"odinger operator with constant magnetic
field.

There is a huge amount of literature dedicated to spectral and statistical
aspects of diamagnetism in large quantum systems. Some of the closely
related papers to our work are \cite{AC, ahs, BC, CoRo, CN, HeSj, If}.

Now let us introduce some notation and give the main theorem.
Consider a magnetic vector potential  $\A= (a_1, a_2,
a_3)=\A_c+\A_p $ where $B \A_c:= \frac{B}{2} {\bf e} \times {\bf x}, \; {\bf e}= (0,0,1)$ is the usual symmetric gauge generated by a
constant magnetic
field $  {\bf B}=   B {\bf e}, B>0 $ and $\A_p$ is $\mathbb{Z}^3$-periodic satisfying
 $\vert \A_p\vert ^2  \in {\cal K}_{loc}(\rt)$.   The notation ${\cal
   K}_{loc}$ denotes the usual Kato class \cite{CFKS,
   Si}. Relations between  these assumptions on magnetic potentials  and  related  choices of periodic magnetic fields are
   discussed  in \cite{If}. Assume that $V$ is also $\mathbb{Z}^3$-periodic such that
$V\in {\cal K}_{loc}(\rt)$. Later on we will give a rigorous sense to
the operator (here $\omega  := e B/c\in \re$)
$$H_\infty(\omega,V):=\frac{1}{2}(-i\nabla-\omega\A)^2+V$$
corresponding to the obvious quadratic form initially defined on
$C^\infty_0(\rt)$. If $\Lambda$ is a bounded open and simply connected
subset of $\rt$ we denote by $H_\Lambda(\omega,V)$ the operator
obtained by restricting the above mentioned quadratic form to
$C^\infty_0(\Lambda)$. We will see that $H_\Lambda(\omega,V)$ has
purely discrete spectrum.

Let $\omega \in \re$ and $\beta := \frac{1}{k_B T} > 0$, where $T>0$
is the temperature and $k_B$ is the Boltzmann constant. Set
$e_0=e_0(\omega)$ to be $\inf\sigma(H_\infty(\omega,V))$.
Introduce the following complex domains
 \bel{D}
D_{+1}(e_0) := {\C} \setminus (-\infty, -e^{\beta
e_0}], \quad D_{-1}(e_0) := {\C} \setminus [e^{\beta e_0},\infty)
\ee

The grand canonical finite volume pressure is defined as \cite{abn1, AC, Hu}
\begin{align}\label{P2}
{P}_{\Lambda}(\beta,\omega, z,\epsilon)
= \frac{\epsilon }{\beta \vert \Lambda \vert }
{\rm Tr}_{L^{2}(\Lambda)} \big
\{\ln( {\mathbb I} + \epsilon z  e^{-\beta H_\Lambda (\omega, V)})\big\}
\end{align}
where $\epsilon = +1$ corresponds to the Fermi case and
$\epsilon= -1$ corresponds to  the Bose case.
In \eqref{P2} the activity $z \in D_\epsilon(e_0) \cap  \re$. The
operator  $\ln( {\mathbb I}+ \epsilon z e^{-\beta H_\Lambda (\omega,
  V)}) $ in the right
hand side of  \eqref{P2}
is defined via functional calculus. Due to some trace class
estimates which we will prove later on (see \eqref{n1W}), the pressure $P_\Lambda$ in
\eqref{P2} is well defined. Define the following complex domains
\begin{equation}
\label{domainDr}
{\bf D}_{\epsilon} := \bigcap_{\omega \in \re} D_{\epsilon}(e_0(\omega))= D_{\epsilon}(e_0(0)),\quad \epsilon =\pm 1 \end{equation}

Now we can formulate  the main result  of this paper:
\begin{theorem}
\label{omanapress2}
Let  $\beta > 0$.

{\rm (i)}. For each open set ${ K}\subset \C$ with the property that $\overline{K}$ is compact and included in $ {\bf
  D}_\epsilon$, there exists a complex
neighborhood ${\cal N}$ of the real axis such that
${\cal N}\,\mathrm{x}\,{ K} \owns (\omega,z) \mapsto P_{\Lambda} (\beta, \omega,z, \epsilon)$
is  analytic.

{\rm (ii)}. Let $\omega \in \re$ and choose a compact set   $K \subset D_{\epsilon}(e_0(\omega))$. Then  uniformly  in $z \in K$
 $$P_{\infty}
(\beta, \omega,z, \epsilon):=\lim_{\Lambda\rightarrow \rt}P_{\Lambda}
(\beta, \omega,z, \epsilon)$$ exists and defines a smooth function of
$\omega$.
\end{theorem}
In  Theorem \ref{omanapress2} $(ii)$,  we take the large volume limit  in the sense defined in the  section 3.4 below.

The rest of the paper contains the proof of this theorem. While {\it (i)}
 and the first part of {\it (ii)} will be proved quite in detail, we only   outline the main ideas leading to the smoothness of $ P_\infty$, all details will be given in \cite{Sa}.

%%%%%%%%%%%%%%%%%%%%%%%%%%%%%%%%%%%%%%%%%%%%%%%%%%%

\section{Technical preliminaries.}
%%%%%%%%%%%%%%%%%%%%%%%%%%%%%%%%%%%%%%%%%%%%%%%%%%%
\label{prelim}

Define
the sesquilinear non-negative form on $C^\infty_0(\Lambda)$ given by
(here $\omega\in \mathbb{R}$):
$$ q_0(\varphi,\psi):=  \frac{1}{2}\langle (-i\nabla - \omega \A)\varphi,
(-i\nabla - \omega \A)\psi \rangle$$
By closing this form we generate a self-adjoint operator denoted by $
H_\Lambda (\omega,0)$, whose form core is $ {C}_0^\infty
(\Lambda)$, see e.g. \cite{BrHuLe,Si}. For convenience we represent this
operator as $H_\Lambda (\omega,0)=\frac{1}{2}(-i\nabla - \omega
\A)^2$.
If $\Lambda=\rt$, the corresponding free   magnetic operator is denoted
by
$ H_\infty(\omega, 0)$. If $\A$ is smooth enough, then $H_\Lambda
(\omega,0)$ can be seen as the Friedrichs extension of $
H_\infty(\omega, 0)$ restricted to $ {C}_0^\infty
(\Lambda)$.

The operator $ H_\Lambda (\omega,0)$ obeys the diamagnetic inequality
\cite{ahs,Si},
\bel{dia}
 \forall \varphi \in L^2(\Lambda),\,\,\forall \beta \geq 0, \quad
\vert e^{ -\beta H_\Lambda (\omega,0)} \varphi  \vert  \leq
e^{ -\beta H_\Lambda (0,0)}   \vert \varphi \vert
 \ee
We will work with electric potentials $V\in {\cal K}_{loc}(\rt)$ which
are $\mathbb{Z}^3$ periodic. We denote the restriction of $V$
to $\Lambda$ by the same symbol.

 It is  known that  $V$ is infinitesimally form bounded to
$H_\Lambda(0,0)$ \cite{CFKS}, and implicitly to $H_\Lambda(\omega, 0)$;
the last statement follows  by  using standard arguments involving the
diamagnetic inequality \eqref{dia} (see \cite{BrHuLe} and references herein).
We conclude that the closure of the sesquilinear form defined on
${C}_0^\infty (\Lambda)$ and given by
 $$ q_V(\varphi,\psi):=  \frac{1}{2} \langle (-i\nabla - \omega\A)\varphi,(-i\nabla
 - \omega \A)\psi \rangle + \langle V\varphi, \psi \rangle  $$
will be symmetric, bounded from below and with the domain
$ Q(q_V) = Q(q_0)$. We denote  by $ H_\Lambda(\omega,V)$ its
associated selfadjoint operator in $ L^2(\Lambda)$.

The diamagnetic inequality  \eqref{dia} holds true if we replace the
free operators by  the perturbed one
$ H_\Lambda(\omega,V)$ and $ H_\Lambda(0,V)$, see e.g. \cite{HuSi}.
This together with the  min-max principle  \cite{RS4} imply:
$$ E_0(\omega) := \inf \sigma (H_\Lambda (\omega, V))\geq e_0(\omega):= \inf \sigma (H_\infty(\omega, V)) \geq  e_0(0)> -\infty.$$

\begin{remark}
The operators  $H_\infty(\omega,0)$ and   $H_\Lambda(\omega,0)$  can
be defined under weaker conditions  on $\A$ see e.g. \cite{BrHuLe, If, HG, Si} but
the  one   imposed here will be necessary in what follows.
When we work with a bounded $\Lambda $, the form domain of
$H_\Lambda(\omega,V)$ will be $\mathcal{H}_0^1(\Lambda)$, independent of $\omega$
and $V$. If $\Lambda=\rt$,
then under our conditions on $V$ and $\A$ the operator $H_\infty(\omega,V)$ is
selfadjoint and  bounded from below  having
$C_0^\infty (\rt)$ as a form core.
\end{remark}

In the rest of the section we only consider the operator defined on
the finite box. We allow $\omega  \in \C$ and want to study the analyticity
properties of the family
$\{ H_{\Lambda}(\omega,V), \omega \in \C \}$.

\begin{proposition} Under conditions given above then $\{ H_{\Lambda}(\omega,V), \omega \in \C \}$ is a type (B) entire family of operators. In particular  $ H_{\Lambda}(\omega,V), \omega \in \C $ are sectorial operators with sector:
\begin{equation} \label{spct}
{\cal S} (\omega):= \{ \xi  \in \C,  \,\,   \vert \Im \xi  \vert   \leq  \vert \omega_1 \vert (c_{1} \Re \xi  + c_{2}  ), \,\, \Re  \xi \in [c_{3} , +\infty) \}
\end{equation}
 where the  constants $c_1,c_2, c_3$ satisfy: $0 < c_{1}, c_{2}< \infty$ and $- \infty < - c_{2}/c_{1} < c_{3} < e_{0}$.
\end{proposition}

\proof Denote by $T$ either $ V$
or $\A^2$. Let  $ \omega_0 \in \re$ and $\varphi \in Q(q_V), \Vert
\varphi\Vert =1$. Then we know that for all $\sigma>0$ there
exists $ \sigma'$  independent of $ \omega_0$ such that
\bel{relb0}
\vert \langle T \varphi, \varphi\rangle   \vert \leq
 \sigma  \langle H_\Lambda(\omega_0,V)\varphi,\varphi\rangle
+ \sigma'
\ee
Let us show that the following two sesquilinear forms
\begin{eqnarray} \label{def-r1-r2}
{r}_{1,\Lambda}(\omega_0) :=  \Re \{ \A\cdot(i\nabla+\omega_0
\A)\}, \quad
{r}_{2,\Lambda} :=  \label{op R}
\frac{1}{2}\;\A^2\,
\end{eqnarray}
are infinitesimally form bounded relatively to the form corresponding
to $H_\Lambda(\omega_0,V)$.

 Let $ \omega_0 \in \re$ and let $\varphi \in Q(q_V), \Vert
 \varphi\Vert =1$.  The Cauchy-Schwarz inequality implies that for any
 $\alpha>0$ we have:
\begin{align}\label{relb1}
\vert \langle \A \varphi, ( i\nabla + \omega_0 \A) \varphi  \rangle
\vert \leq
  \alpha \langle H_\Lambda(\omega_0,0)\varphi,\varphi\rangle +
\alpha^{-1} (\A^2\varphi, \varphi)\nonumber \\
 \leq \alpha  \langle H_\Lambda(\omega_0,V)\varphi,\varphi\rangle +
\alpha \vert  \langle V\varphi, \varphi\rangle \vert
+ \alpha^{-1}  \langle \A^2\varphi, \varphi\rangle
 \end{align}
In view of  \eqref{relb0}, then for all $\vartheta>0$ there exists
$\vartheta' >0$ both $\omega$-independent   such that

\bel{relb2}
\vert \langle \A \varphi, ( i\nabla + \omega_0 \A) \varphi \rangle
\vert  \leq  \vartheta \langle
H_\Lambda(\omega_0,V)\varphi,\varphi\rangle +  \vartheta'
 \ee
This implies that the form ${r}_{1, \Lambda}(\omega_0)$ is bounded
when restricted to the form domain of $H_{\Lambda}(\omega_0,V)$ and
moreover, it generates an operator with  zero relative form bound.
This property also holds for the form ${r}_{2,\Lambda}$.

Now if $\omega \in\C$, denote  $d \omega:= \omega -\omega_0$
and observe that we have in  the form sense:
\bel {def-r} {r}_{\Lambda}(\omega_0,\omega):= d\omega\;
{r}_{1,\Lambda}(\omega_0)  + d\omega^2 {r}_{2,\Lambda},\quad
H_{\Lambda}
(\omega,V)= H_{\Lambda}(\omega_0,V)+{r}_{\Lambda}(\omega_0,\omega)\ee
We conclude  that the form domain of $H_\Lambda(\omega,V)$ is
independent of $\omega$: ${\cal Q}(H_\Lambda(\omega,V))= {\cal
  Q}(H_\Lambda(0,V))$. Notice that  \eqref{def-r}  can be extended for $\omega_0, \omega \in \C$. We will now show that
$\{ H_{\Lambda}(\omega,V), \omega \in \C \}$ is a family of $m$-sectorial operators. Both properties ensure that
$\{ H_{\Lambda}(\omega,V), \omega \in \C \}$ is an analytic family of type (B) (see e.g. \cite{K}).\\

Fix $ \omega \in \C$ with $\Re \omega= \omega_0$, $\Im \omega = \omega_1$ and let
 $\varphi  \in {\cal Q}(H_\Lambda(\omega_0,V)), \; \Vert \varphi \Vert
 =1$.  Using  \eqref{relb0}, we conclude that for all $\sigma>0$ small
 enough such that $\sigma \omega_1^2 \leq 1   $  there exists $
 \sigma'$ such that:
\bel {numr} \Re \langle H_{\Lambda}(\omega,V)  \varphi ,  \varphi \rangle =
 \langle H_{\Lambda}(\omega_{0},V)  \varphi , \varphi \rangle -
\frac{\omega_1^2}{2} \langle \A ^2  \varphi, \varphi \rangle
\geq
(1-   \sigma \omega_1^2/2) \langle H_{\Lambda}(\omega_{0},V)
\varphi,  \varphi \rangle - \frac{\omega_1^2}{2} \sigma'
 \ee
On the other hand,
 from (\ref{relb2}) we conclude that there exist
two constants $\vartheta, \vartheta' > 0$   such that
\begin{equation}\label{numi}
\vert \Im  \langle H_{\Lambda}(\omega,V)  \varphi ,  \varphi \rangle
 \vert =  \vert \omega_1 \Re \langle \A \varphi,
( i\nabla + \omega_0 \A) \varphi \rangle  \vert \leq \vert
\omega_1 \vert  (\vartheta \langle
H_{\Lambda}(\omega_{0},V)  \varphi,  \varphi \rangle + \vartheta' )
\end{equation}
%%%%%%%%%%%%%%
Let  $\Theta(H_{\Lambda}(\omega,V))$ be the numerical range of
$H_{\Lambda}(\omega,V)$.  Then from \eqref{numr} and \eqref{numi} we
obtain that both $\Theta(H_{\Lambda}(\omega,V))$ and
$ \sigma( H_{\Lambda}(\omega,V)) $ are included in the  sector \eqref{spct}.
\qed

\begin{remark} \label{rem1} {\it (i)}. Note that $c_{1}, c_{2}, c_{3}$
in \eqref{spct} depend implicitly on $\omega$ through the condition
$\sigma \omega_1^2 \leq 1 $.  If $\omega_1$ is small enough, then
these constants can be chosen to be $\omega$ independent. Moreover,
let $\omega_0 \in \re$ and $\omega \in \C$ such that
$\vert d \omega\vert $ is small enough. Then  for all
$\varphi \in {\cal Q}(H_{\Lambda}(\omega,V))$, $\Vert \varphi \Vert=1$
$$ \Re \langle H_{\Lambda}(\omega,V)  \varphi ,  \varphi \rangle \geq  c_{3} \geq e_0(\omega_0 ) + \mathcal{O}( \vert d\omega\vert).$$
\medskip
{\it (ii)}. Let $\omega \in \C$.  From \cite{K} we know that if $\xi \notin {\cal S}(\omega)$,
$\Vert (H_{\Lambda}(\omega,V) - \xi)^{-1} \Vert \leq \frac{1}{d(\xi,{{\cal S}})} .$
Hence  put
$\gamma(\omega)= \gamma := - c_{2}/c_{1}$ and $\theta(\omega) = \theta :=  \arctan (c \vert \omega_1\vert)$.  For any  $\delta > 0$ introduce the sector
\begin{equation} \label{sect}
{\cal S}_\delta(\omega):= \{ \xi  \in \C,  \vert \arg( \xi -\gamma) \vert \leq \theta + \delta \}
\end{equation}
Then there exists a constant $c_\delta> 0$ such that for all
$\xi \notin {\cal S}_\delta(\omega)$ we have
\bel{majres}
\Vert (H_{\Lambda}(\omega,V) - \xi)^{-1} \Vert \leq \frac{c_\delta}{\vert \xi - \gamma\vert}.
\ee
{\it (iii)}. The operator $H_\Lambda(0,V)$ has compact resolvent
(see e.g. \cite{GT,RS2}). By standard arguments this also holds true  for
$ H_{\Lambda}(\omega,V),\, \omega \in \C $ \cite{K}.
Hence $ H_\Lambda(\omega,V)$  has only discrete spectrum.
\end{remark}

\vspace{0.5cm}

We are now interested in establishing Hilbert-Schmidt and  trace norm
estimates for powers of the resolvent at finite volume.
Denote by $B_1$ the set of trace norm operators, and by $B_2$
the set of Hilbert-Schmidt operators defined on $L^2(\Lambda)$.
We denote by $\Vert  T \Vert_1$ and $ \Vert T\Vert_2 $ the   trace norm,
respectively  the Hilbert-Schmidt norm of the operator $T$.

For $\beta > 0$ and $\omega \in \re$, let $$ W_\Lambda(\beta, \omega)= W_\Lambda(\beta, \omega,V) :=
e^{-\beta H_\Lambda(\omega,V) }$$ be
the strongly continuous semigroup associated to $H_\Lambda (\omega,V)$
on $L^2(\Lambda)$
see for example  \cite {K, Z} for the definition  and  general properties of
a semigroup.

\begin{lemma} \label{ n12W} There exist two positive constants
$c_0$ and $C_0$ such that for
every $\beta >0$ and $\omega \in \re$ we have that $ W_\Lambda(\beta,
\omega) $ is a positive trace class operators obeying:
 \begin{equation} \label{n1W}
 \Vert W_\Lambda(\beta, \omega) \Vert_1= \rm{Tr}_{L^{2}(\Lambda)}\{W_\Lambda(\beta, \omega)\}  \leq c_0 \beta^{-3/2}e^{C_0 \beta} \vert \Lambda \vert.
\end{equation}
Moreover, its Hilbert-Schmidt   norm satisfies
\begin{equation} \label{n2W}
 \Vert W_\Lambda(\beta, \omega) \Vert_2   \leq c_0 \beta^{-3/4}e^{C_0 \beta} \vert \Lambda \vert^{\frac{1}{2}}.
\end{equation}
\end{lemma}
\begin{proof}
From \cite{BrHuLe} we know that the semigroup is an
integral operator:
$$ (W_\Lambda (\beta, \omega)\varphi)(\x)= \int_\Lambda  G_\Lambda (\x,\y,\beta, \omega)\varphi(\y)d\y, \quad   \varphi \in L^{2}( \Lambda).$$
Moreover the integral kernel $ G_\Lambda$ is   jointly continuous in
$(\x, \y, \beta) \in \Lambda \times \Lambda \times \re_+^*$ and satisfies
\begin{equation} \label{Pnsm}
 \vert G_\Lambda (\x,\y,\beta, \omega) \vert  \leq c_0 \beta ^{-3/2}
e^{C_0\beta } e^{\frac{-\vert \x-\y \vert^2} {4\beta}}, \;
(\x,\y, \beta, \omega) \in \Lambda \times \Lambda \times \re_+^* \times \re
\end{equation}
for some positive constants $c_0,C_0$  which only depend on the
potential
$V$. This result comes from the monotonicity property  of the semigroup and some generalized diamagnetic estimate  \cite{BrLeMu,Si}.

The proof of the lemma follows easily from \eqref{Pnsm}. Note that the use
of   \eqref{Pnsm} is important in order to get the explicit  $\beta$-dependance of quantities involved in the lemma.
\end{proof}

\vspace{0.5cm}

We are now interested in obtaining similar estimates for powers of the
resolvent. Let $ \alpha >0$, $\omega \in \re$, $\xi_0 \in \C, \; \Re
\xi_0 < e_0(\omega)$. As bounded operators on
$ L^2( \Lambda)$ we have  \cite{If, Si}
\bel{rsg}
(H_\Lambda(\omega,V)- \xi_0)^{-\alpha} = \frac {1}{\tilde{\gamma}(\alpha)} \int_0^\infty t^{\alpha-1}e^{\xi_0 t} W_\Lambda(t, \omega) dt
\ee
where $\tilde{\gamma}(\cdot)$ is the  Euler gamma  function. In
particular  from \eqref{n1W} and  \eqref{n2W} for $\Re\xi_0 < 0$ and
$|\Re\xi_0|$ large enough there exists a constant $c > 0$  independent of $\omega \in \mathbb{R}$:
\bel{n2R}
\Vert  (H_\Lambda(\omega,V)- \xi_0)^{-1}\Vert_2 \leq  c \vert \Lambda
\vert^{\frac{1}{2}}\quad {\rm and}\quad \Vert  (H_\Lambda(\omega,V)-
\xi_0)^{-2}\Vert_1 \leq  c \vert \Lambda  \vert.
\ee

%%%%%%%%%%%%%%%%%%%%%%%%%%%%%%%%%%%%%%%%%%%%%%%%%%%%%%%%
\section{Proof of the main theorem.}
%%%%%%%%%%%%%%%%%%%%%%%%%%%%%%%%%%%%%%%%%%%%%%%%%%%%%%%%

%%%%%%%%%%%%%%%%%%%%%%%%%%%%%%%%%%%%%%%%%%%%%%%%%%%%%%%%%%%%%%
\subsection{ $\omega$-analyticity of resolvents.}
%%%%%%%%%%%%%%%%%%%%%%%%%%%%%%%%%%%%%%%%%%%%%%%%%%%%%%%%%%%%%
The first technical result is the following:
\begin{proposition}
\label{thm2}
Let ${\omega} \in \mathbb{C}$ and $\xi \in \rho(H_\Lambda({\omega},V))$.
Then there exists a complex neighborhood ${\cal V}_\xi({\omega})$ of
${\omega}$ such that $\xi \in \rho(H_\Lambda({\omega'},V))$
and the operator valued function
${\cal V}_\xi({\omega}) \owns \omega' \mapsto
(H_\Lambda(\omega',V) -\xi)^{-1}$ is ${ B}_2$-analytic.
\end{proposition}
\begin{proof}  Let ${\omega} \in \mathbb{C}$. First we choose
$\xi_0  < 0 $ negative enough so that  $\xi_0 \in
\rho(H_\Lambda({\omega},V))$. Such a choice is possible
because $H_\Lambda({\omega},V)$ is $m$-sectorial.

 It is a well known fact \cite{RS2} that   since the perturbation
$r_\Lambda(0, \omega)$ (see \eqref{def-r}) is relatively  form bounded
to $H_\Lambda(0,V)$ with zero bound then
for $ \xi_0 <0 $ with $|\xi_0|$ large enough there exists some
complex neighborhood of  ${\omega}$ denoted by $\nu({\omega})$ such
that for all   $\omega' \in \nu({\omega})$ one has:
\bel{er0}
 \Vert (H_\Lambda(0,V)- \xi_0)^{-1/2} r_\Lambda(0, \omega')(H_\Lambda(0,V)- \xi_0)^{-1/2}\Vert <1.
\ee
   Set $ K( \xi_0, \omega'):= (H_\Lambda(0,V)- \xi_0)^{-1/2} r_\Lambda(0, \omega')(H_\Lambda(0,V)- \xi_0)^{-1/2}$.  From the estimate \eqref{er0} we conclude that $\forall \omega' \in \nu({\omega})$,    $\xi_0 \in  \rho(H_\Lambda({\omega'},V))$ and
  \bel{er1}
(H_\Lambda(\omega',V)- \xi_0)^{-1} = (H_\Lambda(0,V)- \xi_0)^{-1/2} ( \mathbb I +  K( \xi_0, \omega'))^{-1}(H_\Lambda(0,V)- \xi_0)^{-1/2}.
\ee
 holds in the bounded operator sense. And since $\omega' \in \nu({\omega}) \to K( \xi_0, \omega')$
is analytic, it follows that the bounded operators valued function
$\omega' \in \nu({\omega}) \to ( \mathbb I +  K( \xi_0, \omega'))^{-1}$ is analytic too.

On the other hand, from \eqref{er1}  we have

 $$ \Vert(H_\Lambda(\omega',V)- \xi_0)^{-1} \Vert_2 \leq   \Vert( \mathbb I +  K( \xi_0, \omega'))^{-1}\Vert\Vert(H_\Lambda(0,V)- \xi_0)^{-1} \Vert_2$$
which together with  \eqref{n2R}, it shows that  $\omega' \in
\nu({\omega}) \to (H_\Lambda(\omega',V)- \xi_0)^{-1} $ is a
Hilbert-Schmidt family of operators in $\omega' \in
\nu({\omega})$. Now it is straightforward to prove the theorem for
such a $\xi_0$. It remains to extend the ${ B}_2$-analyticity property
for any $\xi \in \rho(H_\Lambda({\omega},V))$.

 Let $\xi_0$ as above and consider the first resolvent equation
\begin{equation}
\label{fieqres}
(H_\Lambda({\omega},V)- \xi)^{-1} = (H_\Lambda({\omega},V)- \xi_0)^{-1} +(\xi - \xi_0) (H_\Lambda({\omega},V)- \xi)^{-1}(H_\Lambda({\omega},V)- \xi_0)^{-1}
\end{equation}
Since there exists a bounded complex neighborhood $V_\xi({\omega})$
of ${\omega}$ such that the operator-valued function $V_\xi({\omega})
\owns \omega' \mapsto (H_\Lambda(\omega',V)- \xi)^{-1}$ is
bounded-analytic, by standard arguments involving the bilateral ideal
property of $B_{2}$, the operator-valued function
$\omega' \mapsto (H_\Lambda({\omega'},V)-
\xi)^{-1}(H_\Lambda({\omega'},V)- \xi_0)^{-1}$ is ${ B}_2$-analytic on
$V_\xi({\omega})\cap \nu({\omega})$. Now use \eqref{fieqres} and
the proof is over.
\end{proof}

\vspace{0.5cm}

\begin{follow}
\label{coro}
Let ${\omega} \in \C$ and $\xi \in \rho(H_\Lambda({\omega},V))$.
Then there exists a neighborhood ${\cal V}_\xi({\omega})$ of ${\omega}$ such that the operator valued function ${\cal V}_\xi({\omega}) \owns \omega' \mapsto  (H_\Lambda(\omega',V) -\xi)^{-2}$ is ${ B}_1$-analytic.
\end{follow}

\begin{proof} From Proposition \ref{thm2} we have that
$(H_\Lambda(\omega,V) -\xi)^{-2}$ is a  product of two Hilbert-Schmidt
operators. Thus $(H_\Lambda(\omega,V) -\xi)^{-2}$ is trace class. Then
by a  direct  calculus we can check the statement  of the corollary.
\end{proof}

\vspace{0.5cm}

Now we consider  $W_\Lambda(\beta,\omega), \, \omega \in \mathbb{R}$. We
want to extend $W_\Lambda(\beta,\omega)$ to complex $\omega$'s and in
trace class sense. We will use the fact that
the operator $H_\Lambda(\omega,V)$  is $m$-sectorial:
\begin{follow}
\label{coro1}
Let $ \beta >0 $. The family  $\{ W_\Lambda(\beta,\omega), \omega \in
\mathbb{R}\}$ can be extended to a $B_1$ entire family of operators. \end{follow}
\begin{proof} Let $\beta > 0$, $\omega \in \mathbb{C}$.   Consider the curve in
  $\C$ given by  $ \Gamma :=\{ \xi \in \C\,:\,  \vert \arg( \xi -\gamma') \vert = \theta
  + \epsilon \}$ where $ \gamma', \epsilon$ are chosen such that
$\gamma- \gamma' =1$ and
$\theta + \epsilon < \frac{\pi}{2}$. Here $ \gamma, \theta $ are given by the Remark \ref{rem1} ii). The curve $\Gamma$ encloses the spectrum of $H_\Lambda(\omega',V)$ for all $\omega'$ in a neighborhood of $\omega$, $\nu(\omega)$.
From the Dunford functional calculus \cite{DS},  the following relation holds in terms of bounded operators:
\begin{equation}
\label{depart}
 W_\Lambda(\beta,\omega') :=  \frac{i}{2 \pi}\int_\Gamma d\xi e^{-\beta \xi}
( H_\Lambda(\omega',V)-\xi)^{-1}, \; \quad \omega' \in \nu(\omega)
\end{equation}
This shows that the semigroup can be extended to a bounded operator for every complex $\omega'$.
We now want to show that this formula defines
in fact a trace class operator.
Choose $\xi_0 \in \C$ with $\Re\xi_0 <0$ and $|\Re\xi_0|$ large enough
so that $\xi_0 \in \rho( H_\Lambda(\omega',V))$ for any $\omega' \in
\nu(\omega)$. Using twice the resolvent formula \eqref{fieqres} in
\eqref{depart} we obtain the identity:
\begin{equation}
\label{depart1}
 W_\Lambda(\beta,\omega') = \frac{i}{2 \pi} \int_\Gamma d\xi e^{-\beta \xi} (\xi-\xi_0)^2( H_\Lambda(\omega',V)-\xi)^{-1}
 ( H_\Lambda(\omega',V)-\xi_0)^{-2}, \quad \omega' \in \nu(\omega)
\end{equation}
From the choice of $\Gamma$, the bounded operator valued
function $\nu(\omega)\ni \omega' \to ( H_\Lambda(\omega',V)-\xi)^{-1}$
is analytic for all $\xi \in \Gamma$, and all norm bounds are uniform in
$\xi\in\Gamma$. Moreover, from \eqref{majres} we conclude that
there exists a constant $C>0$ such that
$\Vert e^{-\beta \xi} (\xi-\xi_0)^2( H_\Lambda(\omega',V)-\xi)^{-1}
 \Vert  \leq C \vert \Re \xi  \vert^2 e^{-\beta \Re \xi}$. Therefore
 $$ \nu(\omega)\ni \omega'  \to \int_\Gamma d\xi e^{-\beta \xi} (\xi-\xi_0)^2( H_\Lambda(\omega',V)-\xi)^{-1}$$
is bounded analytic too.
Hence from  Corollary  \ref{coro}  and  \eqref{depart1} the
operators valued function $\nu(\omega)\ni\omega'
\to W_\Lambda(\beta,\omega') $ is $B_1$-analytic. Thus $W_\Lambda(\beta,\cdot\,)$ is $B_{1}$-entire.
\end{proof}

%%%%%%%%%%%%%%%%%%%%%%%%%%%%%%%%%%%%%%%%%%%%%%%%%%%%%%%%%%
\subsection{$\omega$-analyticity of the pressure.}
%%%%%%%%%%%%%%%%%%%%%%%%%%%%%%%%%%%%%%%%%%%%%%%%%%%%%%%%%%%%
\medskip

Let    $\beta>0$, $\omega \in \re$  and  $z \in D_\epsilon(e_{0}) \cap  \re$. Define
\bel {f}
 [e_0(\omega) , \infty) \owns \xi  \mapsto
\ln {\big(1 + \epsilon z e^{-\beta \xi}} \big)\nonumber
\ee

We have that the map  $(z,\xi)\mapsto \ln {\big(1 + \epsilon z e^{-\beta \xi}} \big)$ is jointly analytic in
\begin{equation}
\label{proana1}
\big\{ (z, \xi)\in \C \times \C\,:\,\vert z \vert e^{- \beta \Re \xi} <1\big\}
\end{equation}
but this is not sufficient and we also need to control the region in which
$\Re \xi$ is close to the bottom of the spectrum. Let $\beta>0$, $\omega \in \re$, $-\infty < e'_0 \leq  e_0= e_0(\omega)$  and consider the domains $D_\epsilon  (e'_0)$ defined as in \eqref {D}  but  with $e'_0$ instead $e_0$. Then:

\begin{lemma} \label{anf1} Let   $\beta>0$, $\omega \in \re$ and  $-\infty < e'_0 \leq e_0$. For each compact $K \subset D_\epsilon(e'_0)$ there exists $\eta_{K} >0$
such that  $(z,\xi)\mapsto \ln {\big(1 + \epsilon z e^{-\beta \xi}} \big)$
is  jointly analytic in

\begin{equation} \label{etaK} K  \times \Big\{ \xi \in \mathbb{C}\,:\,  \Im \xi \in \Big(- \frac{\eta_{K}}{\beta}, \frac{\eta_{K}}{\beta}\Big), \quad \Re \xi \in [e'_0, \infty) \Big\}
\end{equation}
If $K'$ is a compact subset such that $ K' \subset K$ then $\eta_{K}' > \eta_{K}$.
\end{lemma}
\begin{proof}

We first deal with the Bose case. Here ${\cal B}(r)$ is an open ball in $\C$
centered at the origin having radius $r>0$.  Obviously
$(z,\xi ) \mapsto \ln {\big(1 + \epsilon z e^{-\beta \xi}} \big)$ is a  jointly analytic  function in $ \Re \xi \in [e'_0(\omega), \infty), z \in {\cal B}( e^ {\beta e'_0})$.
Let $K$ be  a compact of $D_{-1}(e'_0)$ and denote by $\tilde K = K \setminus {\cal B}( e^ {\beta e'_0})$. Let
$$ \theta_m := \inf\{ \arg(z), z \in \tilde K\}, \quad  \theta_M := \sup\{ \arg(z), z \in \tilde K\}$$
Because  ${\rm dist} (\tilde K, [e^{\beta e'_0}, \infty) ) >0$, then $0< \theta_m \leq \theta_M < 2\pi$. We set $\eta_K:=
\frac{1}{2}\inf \{\theta_m, 2\pi-\theta_M \}$. Clearly for $z \in \tilde K$ and  $\Im \xi \in [- \frac{\eta_K}{\beta},\frac{\eta_K}{\beta}]$,
$0<\frac{\theta_m}{2}  \leq \arg z - \beta\Im \xi \leq \pi + \frac{\theta_M}{2} < 2\pi$ then $ \Im(1 -z e^{-\beta \xi})= 0$ iff $ \arg z - \beta\Im \xi = \pi$ but in this last case $\Re(1 -z e^{-\beta \xi}) >0 $.\\

For the Fermi case we get the lemma following the same arguments as
above.
Let $K$ be  a compact of $D_{+1}(e'_0)$ and denote by $\tilde K = K \setminus {\cal B}( e^ {\beta e'_0})$. Let
$$ \theta_m := \sup \{ \arg(z), z \in \tilde K, \arg(z)\geq 0 \}, \quad  \theta_M := \inf \{ \arg(z), z \in \tilde K, \arg(z) <0\}$$
We set $\eta_{K} := \frac{1}{2} \inf \{\pi - \theta_m, \pi + \theta_M\}$. Clearly for $z \in \tilde K$ and  $\Im \xi \in [- \frac{\eta_K}{\beta},\frac{\eta_K}{\beta}]$,
$-\pi<-\frac{\pi}{2} + \frac{\theta_M}{2}  \leq \arg z - \beta\Im \xi \leq \frac{\pi}{2} + \frac{\theta_m}{2} < \pi$ then $ \Im(1 +z e^{-\beta \xi})= 0$ iff $ \arg z - \beta\Im \xi = 0$ but in this last case $\Re(1 +z e^{-\beta \xi}) >0 $.\\
\end{proof}

\vspace{0.5cm}

\begin{proposition}\label{omnapress0}
Let $\beta > 0$, $\omega_0 \in \re$ and $ K \subset D_\epsilon(e_0(\omega_0)) $ a compact subset.
Then there exists a complex neighborhood ${\cal V}(\omega_0)$ of $\omega_0$ such that for any $z \in K$, the pressure is an analytic function w.r.t. $\omega$ in ${\cal V}(\omega_0)$.
\end{proposition}
\begin{proof}

Let $\omega_0 \in \re$,
 and $ K \subset D_\epsilon(e_0) $,  $e_0= e_0(\omega_0)$ be  a compact subset.Then there exists  $e'_0$ satisfying $- \infty < e'_0< e_0$ such that $K \subset D_\epsilon(e'_0)$.
Consider  now  the following positively oriented contour defined by

$$
\Gamma_{K} := \Big\{ \Re\xi = e'_0,\,\,\Im \xi \in \Big[-\frac{\eta_{K}}{2 \beta},\frac{\eta_{K}}{2 \beta}\Big]\Big\} \cup \Big\{ \Re \xi \in [e'_0, \xi_K), \,\,\vert \Im \xi  \vert=  \frac{\eta_{K}}{2 \beta}\Big\} \cup$$

\begin{equation}\label{contour}
 \Big \{  \Re \xi  \geq  \xi_K ,  \,\,  \arg\Big( \xi -   \xi_K \mp i  \frac{ \eta_K }{2\beta}\Big) = \mp  \frac{\pi}{4} \Big \}
 \end{equation}
where $ \eta_K > 0$ is given  by  \eqref{etaK}; $\xi_K$ is chosen so that $\xi_{K} > e_{0}$ and satisfies the condition \eqref{proana1} i.e.
$$ \sup_{z \in K} \{ \vert z \vert \} e^{- \beta \Re \xi} <1 \quad \rm {if}  \quad  \Re \xi > \xi_K $$

Recall  the domain of analyticity of $\xi \mapsto \ln {\big(1 + \epsilon z
  e^{-\beta \xi}} \big)$ defined by lemma \ref{anf1} knowing
\eqref{proana1}. Then $\Gamma_{K}$ is enclosed in this domain of
analyticity.

Let ${\cal B}(\omega_0,r)$  be an open ball in $\C$ centered at
$\omega_0$ and radius $r > 0$. If  $r$ is small enough
then for $\omega \in  {\cal B}(\omega_0,r)$, the spectrum of $H_\Lambda(\omega,V)$ as well as the sector ${\cal S}(\omega)$ defined in \eqref{sect} for $\Re \xi > \xi_{K}$ together lie
inside $\Gamma_{K}$. To see this  we use the Remark \ref{rem1} i).

 For $\beta >0$, $z \in K$ and $\omega \in {\cal B}(\omega_0,r) $ consider the following  Dunford  integral operator \cite{DS}
 $$ I(\beta,z, \omega):= \frac{i}{2\pi}\int_{\Gamma_{K}}
 \mathrm{d}\xi\,
\ln {\big(1 + \epsilon z e^{-\beta \xi}} \big)\big(H_{\Lambda}(\omega,V) - \xi\big)^{-1} $$
The above  integral
  converges and defines a bounded operator
due to the exponential decay of $\ln {\big(1 + \epsilon z e^{-\beta \xi}} \big)$ in $\Re \xi$ and because of \eqref{majres}.

Again here the  choice of the contour implies that if  $r$
is small enough then  for each
   $\xi \in \Gamma_{K}$  the bounded  operator valued function $ {\cal
     B}(\omega_0,r)   \owns \omega \mapsto
   (H_{\Lambda}(\omega,V)-\xi\big)^{-1}$ is analytic.
Therefore for $r$ small enough
  $\{ I(\beta,z, \omega),\, \omega \in  {\cal B}(\omega_0,r)\}$  is an
  analytic family of  bounded operators  in  $L^2(\Lambda)$. By
  analytic continuation we conclude that
$$I(\beta,z, \omega)= \ln{\big( {\mathbb I} + \epsilon
  zW_{\Lambda}(\beta,\omega)\big)}  $$ for all $\omega \in  {\cal
  B}(\omega_0,r)$ because the equality holds for real $\omega$.

Now choose a $\xi_0$ with a very negative $\Re \xi_0 $.
Then we get:
\begin{equation}
\label{lnomco}
\ln{\big({\mathbb I} + \epsilon z W_{\Lambda}(\beta,\omega)\big)} =
\bigg(\frac{i}{2\pi}\int_{\Gamma_{K}} \mathrm{d}\xi\, (\xi -
\xi_0)^{2}
\ln {\big(1 + \epsilon z e^{-\beta \xi}} \big)
(H_{\Lambda}(\omega,V)-\xi\big)^{-1}\bigg) \big(H_{\Lambda}(\omega,V) -\xi_0\big)^{-2}
\end{equation}
This implies that if $r$ is small enough, the family
 $\{ I(\beta,z, \omega),\, \omega \in  {\cal B}(\omega_0,r)\}$
is also analytic in the trace class
 topology.
The proof is over.
\end{proof}

\vspace{0.5cm}
%%%%%%%%%%%%%%%%%%%%%%%%%%%%%%%%%%%%%%%%%%%%%%%%%
\subsection{Proof of the analyticity w.r.t $\omega$ and $z$.}
%%%%%%%%%%%%%%%%%%%%%%%%%%%%%%%%%%%%%%%%%%%%%%%%%%%
Recall that  the Hartog  theorem \cite{hartog},  implies the joint analyticity once we know the analyticity  w.r.t. each variable separately.

Put ${\cal V} := \bigcup_{\omega_0 \in \re} {\cal V}(\omega_0)$. Then
from the Proposition \ref{omnapress0} for any $z \in { K}$ the pressure is an
analytic function
with respect to $\omega$ in ${\cal V}$.
This is the  first thing we need in order to apply the Hartog theorem.

 Now let $ \beta > 0$ and $K$ as in the theorem.  We want to show that
there exists a neighborhood  of the real axis $ {\cal N}$ such that
for any $\omega \in {\cal N}$, the function $ K \owns z \mapsto P_{\Lambda}
(\beta, \omega,z,\epsilon) $ is analytic.

We use formula \eqref{lnomco} but with $e_{0}^{'} < e_{0}(0)$ in the definition \eqref{contour} of $\Gamma_{K}$. The only thing we have   to show is  that
for $\omega \in \C$, $  \Im \omega $ small enough  the  $ { B}_1 $-operator  valued function  $ \Gamma_{K}\ni\xi \to  (H_{\Lambda}(\omega)-\xi\big)^{-1} \big(H_{\Lambda}(\omega) -\xi_0\big)^{-2}$ is  uniformly bounded for $ \Re \xi$ large enough.
But   this is  true since  \eqref{majres}  implies  that   $ \xi \in
\Gamma_{K} \to  \Vert (H_{\Lambda}(\omega)-\xi\big)^{-1} \Vert$  is
uniformly bounded for $ \Re \xi$ large enough and we know that $
\big(H_{\Lambda}(\omega) -\xi_0\big)^{-2} \in B_2$. The proof is
over.
\qed

\vspace{0.5cm}

For $\beta > 0$, $\omega_0 \in \re$ and $z \in D_\epsilon(e_0)$, the grand canonical generalized susceptibilities at finite volume are defined by
\begin{equation}
\label{suscepti}
\chi_\Lambda^N(\beta,\omega_0, z,\epsilon) := \bigg(\frac{e}{c}\bigg)^{N}\frac{\partial^N P_\Lambda}{\partial \omega^N}(\beta,\omega_0,z,\epsilon),\quad N \in \mathbb{N}^*
\end{equation}

By Proposition \ref{omnapress0}, $\chi_\Lambda^N(\beta,\omega_0,z,\epsilon) $   are well defined.
In the physical  literature (see e.g. \cite{Hu}), the cases $N=1$ and $N=2$ correspond respectively to the grand canonical magnetization and  magnetic susceptibility per unit volume. Moreover

\begin{follow} Let $\beta > 0$ and  $N \geq 1$. For each open set $K$ with the property that $\overline{K}$ is compact and $\overline{K} \subset {\bf D}_\epsilon$,  there exists a complex neighborhood ${\cal N}$ of the real axis such that ${\cal N}\,\mathrm{x}\,{K} \owns (\omega,z) \mapsto \chi_\Lambda^N(\beta,\omega,z,\epsilon)$  is  analytic.
\end{follow}

%%%%%%%%%%%%%%%%%%%%%%%%%%%%%%%%%%%%%%
\subsection{The thermodynamic limit.}
%%%%%%%%%%%%%%%%%%%%%%%%%%%%%%%%%%%%%%%%%%%%%
Now assume that the domain $\Lambda$ is obtained by dilating a given set $\Lambda_1\subset \mathbb{R}^3$ which is supposed to be bounded, open, simply connected and with smooth boundary. More precisely:
$$\Lambda_L:=\{\x\in \re ^3:\; \x/L\in \Lambda_1,\quad L>1\}.$$

Assume that the electric potential $V$ belongs to $\mathcal{K}_{loc}$ and
is $\mathbb{Z}^3$ periodic, and
denote its elementary cell with $\Omega$. We also assume that the
magnetic potential $\A$ can be written as $\A_c+\A_p$, where $\A_c$
is the symmetric  gauge given by a constant magnetic field (thus has a
linear growth), while $|\A_p|^2$ belongs to $\mathcal{K}_{loc}$ and is
$\mathbb{Z}^3$ periodic. Let
$\chi_\Omega$ denote the characteristic function of the elementary
cell.

Let $\omega \in \re$. Introduce the integrated density of states (IDS)  defined  as the  following thermodynamic limit if it exists. Let $E \in \re$ and $ N(E)$ denotes the number of eigenvalues of the operator $H_\Lambda (\omega,V)$ smaller than $E$ counting with their multiplicity.
\begin{equation}
\label{ids0}
\rho(E) :=  \lim_{L\to \infty} \frac{N(E)} { \vert \Lambda_L\vert }.
\end{equation}
Let  $P(I)$  the spectral projector associated with the operator $H_\infty(\omega,V)$ on the interval $I$. Then we have \cite{DoIwMi, If}
\begin{proposition} \label{idsp}
Under the condition stated above then the IDS of $H_\infty(\omega,V)$  exists  and for almost all
$E \in \re$ \begin{equation}
\label{ids}
\rho(E) :=  \frac{1}{\vert  \Omega\vert} \rm{Tr}_{L^{2}(\mathbb{R}^{3})}( \chi_\Omega \rm{P}( E))
\end{equation}
where $\rm{P}(E):= \rm{P}(-\infty,E]$.
\end{proposition}

This last proposition allows  us to compute the thermodynamic limit of the pressure. Recall that
we have shown that the pressure at finite volume  for $ \omega \in \re, \beta >0,$ and $ z \in K$
where $K$ is a compact subset of  $ D_\epsilon(e_0)$ can be expressed as:
\begin{align}\label{april10}
P_{\Lambda_L}(\beta,\omega,z,\epsilon)=\frac{i\epsilon}{2\beta\pi |\Lambda_L|}{\rm Tr}_{L^2(\Lambda_L)}\int_{\Gamma_{K}}
 \mathrm{d}\xi\,
\ln {\big(1 + \epsilon z e^{-\beta \xi}} \big)\big(H_{\Lambda_L}(\omega,V) - \xi\big)^{-1}
\end{align}

Let $\omega$  real,  define  \cite{CN2} :
\begin{align}\label{april1}
P_\infty(\beta,\omega,z,\epsilon)
&=\frac{i\epsilon}{2\beta\pi\vert \Omega \vert}{\rm Tr}_{L^2(\re^3)}\int_{\Gamma_{K}}
 \mathrm{d}\xi\,
\ln {\big(1 + \epsilon z e^{-\beta \xi}}
\big)\chi_\Omega\big(H_{\infty}(\omega,V) -
\xi\big)^{-1}\chi_\Omega \nonumber
\end{align}

The above integral defines a trace class operator on $L^2(\re^3)$
because after a use of the resolvent identity we can change the integrand into:
$$(\xi-\xi_0)\ln {\big(1 + \epsilon z e^{-\beta \xi}} \big)\chi_\Omega\big(H_{\infty}(\omega,V) - \xi\big)^{-1}\big(H_{\infty}(\omega,V) - \xi_0\big)^{-1}\chi_\Omega$$
where $\xi_0$ is some fixed and negative enough number. Using the Laplace transform and the properties of the semigroup $e^{-tH_\infty}$ one can prove that $\chi_\Omega\big(H_{\infty}(\omega,V) - \xi\big)^{-1}$ and $\big(H_{\infty}(\omega,V) - \xi_0\big)^{-1}\chi_\Omega$ are Hilbert-Schmidt operators whose norms grow polynomially with $\Re\xi$.

Then  we have:

\begin{theorem} \label{LT}
Let $\omega \in \re$, $ \beta >0$ and $K \subset D_\epsilon(e_0)$ a compact set. Under the  same condition as  above then uniformly in $ z \in K$
 \begin{equation} \label{april1}
\lim_{L\to\infty}P_{\Lambda_L}(\beta,\omega,z,\epsilon) = P_\infty(\beta,\omega,z,\epsilon).
\end{equation}
\end{theorem}

\proof Define $ (\xi,z) \mapsto f( \xi, z)= f(\xi,\beta, z, \epsilon) := \ln{\big(1 + \epsilon z e^{-\beta \xi}} \big)$. First recall the well known formula (see e.g. \cite{BrCoZa}) which holds if the IDS exists
$$ \lim_{L\to\infty}P_{\Lambda_L}(\beta,\omega,z,\epsilon) = P_\infty(\beta,\omega,z,\epsilon); \quad
P_\infty(\beta,\omega,z,\epsilon) =- \frac{\epsilon}{\beta} \int _{\re} f'_\xi (\xi,z) \rho(\xi) d\xi $$
Then by using  Proposition \ref{idsp} and the fact that $ \{\chi_\Omega \rm{P}(E), E \in \re\}$ is a family of trace class operators we get
$$P_\infty(\beta,\omega,z,\epsilon) =-
\frac{\epsilon}{\beta\vert  \Omega\vert} \int _{\re} f'_\xi (\xi,z)  {\rm Tr}_{L^2(\re^3)}( \chi_\Omega \rm{P}(\xi))d\xi  =
 \frac{\epsilon}{\beta\vert  \Omega\vert} {\rm Tr}_{L^2(\re^3)}\bigg(\chi_\Omega \int _{\re} f(\xi,z)  d\rm{P}(\xi)\bigg) .$$
 So by the spectral theorem
$$P_\infty(\beta,\omega,z,\epsilon) =  \frac{\epsilon}{\beta\vert  \Omega\vert}  {\rm Tr}_{L^2(\re^3)}( \chi_\Omega f(H_\infty,z))$$
and then by using again the Dunford integral representation \cite{DS} the theorem follows. \qed
\newline

The fact that $\omega$ must be real is an important ingredient of the
proof of \eqref{april1} where one extensively uses the gauge
invariance of the operators and the fact that $H_\infty$ commutes with
the magnetic translations generated by $\mathbb{Z}^3$. It is shown in \cite{Sa} that if $\A_c=0$ i.e. the magnetic vector potential is periodic, then  the limit in \eqref{april1} holds true for every $\omega$ is a small ball around every $\omega_0\in\re$, provided that $z$ and $\beta$ are fixed. The explanation is that the analyticity ball in $\omega$ which we have constructed for each $P_{\Lambda_L}$ would be independent of $L$. If $\mathcal{C}_r(\omega_0)$ denotes the positively oriented circle with radius $r$ and center at $\omega_0$, then for any real $\omega$ inside $\mathcal{C}_r(\omega_0)$ and for $r$ small enough we can write:
$$P_{\Lambda_L}(\omega)=\frac{1}{2\pi i}\int_{\mathcal{C}_r(\omega_0)}\frac{P_{\Lambda_L}(\omega')}{\omega'-\omega}d\omega',\quad
\chi_{\Lambda_L}^N(\omega)=\frac{N!}{2\pi i}\int_{\mathcal{C}_r(\omega_0)}\frac{P_{\Lambda_L}(\omega')}{(\omega'-\omega)^{N+1}}d\omega'$$
The last integral representation of $\chi_{\Lambda_L}^N(\omega)$ tells us
that if the pressure admits the thermodynamic limit, the same property
holds true for all generalized susceptibilities.
Thus the existence of the thermodynamic limit of the generalized susceptibilities
follows easily if there is no linear growth in the magnetic
potential generated by the magnetic field.

If $\A_c$ is not zero, then the above argument breaks down because $r$ (the
analyticity radius in $\omega$ of $P_{\Lambda_L}$) goes to zero with
$L$. In fact one cannot hope to prove in general that $P_\infty$ is
real analytic in $\omega$, although one can prove that it  is smooth in
$\omega\in\re$. In order to achieve that, one needs to use
the magnetic perturbation
theory methods developed in \cite{CN2, BrCoLo1, BrCoLo2,
  BrCoLo3}. Complete proofs will be given in \cite{Sa}.
%%%%%%%%%%%%%%%%%%%%%%%%%%%%%%%%%%%%%%%%%%%%%
\subsection{The canonical ensemble.}
%%%%%%%%%%%%%%%%%%%%%%%%%%%%%%%%%%%%%%%%%%%%%
Let $\rho_{0} > 0$ be the density of particles. The number of particles in the finite box $\Lambda$ is $N_{\Lambda} = \rho_0\vert \Lambda \vert$.
For $\beta > 0$, $\omega_{0} \in \re$ and $\rho_{0} > 0$ fixed, define the finite volume Helmholtz free energy (see \cite{Hu}) as
\begin{equation}
\label{freef'}
f_{\Lambda}(\beta,\rho_{0},\omega_{0},\epsilon) := - \frac{1}{\beta} \ln \big( Z_{\Lambda}(\beta,\rho_{0},\omega_{0},\epsilon)\big)
\end{equation}
where $Z_{\Lambda}(\beta,\rho_{0},\omega_{0}) > 0$ stands for the canonical partition function.\\

As a   consequence of  Theorem \ref{omanapress2}, we have:
\begin{follow} \label{fre} Let $\beta > 0$ and  $\rho_{0} > 0$. Then there exists a complex neighborhood of the real axis $\mathcal{M}$ such that the the Helmholtz free energy $\mathcal{M} \owns \omega \mapsto f_{\Lambda}(\beta,\rho_{0},\omega,\epsilon)$ is analytic.
\end{follow}

\begin{proof}
For all $\omega_{0} \in \re$, the canonical partition function is related to the grand-canonical pressure by (see \cite{C})
\begin{equation}
\label{freef}
Z_{\Lambda}(\beta,\rho_{0},\omega_{0},\epsilon) := \frac{1}{2i\pi} \int_{\mathcal{C}} dz \frac{1}{z} \bigg[\frac{\exp\big(\frac{\beta}{\rho_{0}} P_{\Lambda}(\beta,\omega_{0},z,\epsilon)\big)}{z}\bigg]^{N_{\Lambda}}
\end{equation}
where $\mathcal{C}$ is a closed contour around $0$ and included in the analyticity domain ${\bf D}_{\epsilon}$ of  the function $z \to P_{\Lambda}(\beta,\omega_{0},z,\epsilon)$. From  Theorem \ref{omanapress2}, there exists a complex neighborhood ${\cal M'}$ of the real axis such that ${\cal M'} \owns \omega \mapsto Z_\Lambda(\beta,\rho_{0},\omega,\epsilon)$ is analytic. Since $Z_\Lambda(\beta,\rho_{0},\omega_{0},\epsilon) > 0$ for all $\omega_{0} \in \re$, then by a continuity argument, there exists a complex neighborhood ${\cal M}$ of the real axis such that for all $\omega \in {\cal M}$, $\Re Z_\Lambda(\beta,\rho_{0},\omega,\epsilon) > 0$. Then the corollary follows.
\end{proof}

For $\beta > 0$, $\rho_{0} > 0$ and $\omega_0 \in \re$, the canonical generalized susceptibilities at finite volume are defined by
\begin{equation}
\label{susceptic}
m_\Lambda^N(\beta,\rho_{0},\omega_0,\epsilon) := -\frac{1}{\vert \Lambda \vert} \bigg(\frac{e}{c}\bigg)^{N}\frac{\partial^N f_\Lambda}{\partial \omega^N}(\beta,\rho_{0},\omega_0,\epsilon),\quad N \in \mathbb{N}^*
\end{equation}

By   Corollary \ref{fre}, $m_\Lambda^N(\beta,\rho_{0},\omega_0,\epsilon) $ are well defined. Moreover:

\begin{follow} Let $\beta > 0$, $\rho_{0} > 0$ and $N \geq 1$. Then there exists a complex neighborhood ${\cal M}$ of the real axis such that ${\cal M} \owns \omega \mapsto m_\Lambda^N(\beta,\rho_{0},\omega,\epsilon)$ is analytic.
\end{follow}

\vspace{0.5cm}

\noindent{\bf Acknowledgments}. This paper is dedicated to the memory of our colleague and friend Pierre Duclos (1948-2010). Part of this work has been done while B.S. was visiting Aalborg. H.C. acknowledges partial support from
the Danish F.N.U. grant {\it Mathematical Physics}.


\begin{thebibliography} {[50]}
\frenchspacing \baselineskip=12 pt plus 1pt minus 1pt

\bibitem{abn1} {\sc N. Angelescu, G. Nenciu, M. Bundaru},
{\em On the Landau  diamagnetism}.
Comm. Math. Phys.  {\bf 42}  (1975), 9-28.

\bibitem{abn2} {\sc N. Angelescu, G. Nenciu, M. Bundaru},
{\em On the perturbation of Gibbs semigroups.}  Comm. Math. Phys.  {\bf 42}  (1975), 29-30.

\bibitem{AC} {\sc N. Angelescu, A. Corciovei}, {\em On free quantum gases
 in a homogeneous magnetic field.}
Rev. Roum. Phys. {\bf 20} (1975), 661-671

\bibitem{ahs} {\sc J. Avron, I. Herbst, B. Simon},
{\em Schr\"{o}dinger operators with  magnetic  fields.  I.  General
interactions}. Duke Math. J. {\bf 235} (1978), 847-883.


\bibitem{BrHuLe} {\sc K. Broderix, D. Hundertmark, H. Leschke}, {\em Continuity properties of Schr\" odinger Semigroup with magnetic fields}.   Rev. Math. Phys.  {\bf 12}  (2000),  no. 2, 181-225.

\bibitem{BrLeMu} {\sc K. Broderix, H. Leschke, P. M\"uller}, {\em Continuous integral kernels for unbounded Schr\"odinger semigroups and their spectral projections.} J. Funct. Anal.  {\bf 212}  (2004),  no. 2, 287-323.

\bibitem {BrCoLo1} {\sc Ph. Briet, H.D. Cornean, D. Louis}, {\em Diamagnetic expansions for perfect quantum gases.}  J. Math. Phys.  {\bf 47}  (2006),  no. 8, 083511, 25 pp.

\bibitem{BrCoLo2} { \sc Ph. Briet, H.D. Cornean, D. Louis}, {\em  Generalized susceptibilities
 for a perfect quantum gas}. Markov Process. Related Fields, {\bf
 11} (2005), 177-188.

\bibitem{BrCoLo3} { \sc Ph. Briet, H.D. Cornean, D. Louis}, {\em Diamagnetic expansions for perfect quantum gases II: uniform
    bounds.} Asymptotic Analysis  {\bf 59}(1-2) (2008), 109-123.

\bibitem{BC}  { \sc Ph. Briet, H.D. Cornean}, {\em Locating the spectrum for
 magnetic Schr\"odinger and Dirac operators.}
Comm. Partial
 Differential Equations {\bf 27} (2002), no. 5-6, 1079-1101.

\bibitem{BrCoZa}  { \sc Ph. Briet, H.D. Cornean,  V. Zagrebnov}, {\em  Do bosons condense in a homogeneous magnetic field?},   J. Statist. Phys.  116  (2004),  no. 5-6, 1545--1578.

\bibitem{CoRo}  { \sc M. Combescure, D. Robert}, {\em Rigorous semiclassical results for the  magnetic
response of an electron gas.} Rev. Math. Phys. {\bf 13} (2001), 1055-1073


\bibitem{C} {\sc H.D. Cornean} {\em  On the magnetization of a charged Bose
 gas in the canonical ensemble.} Commun. Math. Phys. {\bf
 212} (2000), 1-27

\bibitem{CN}  { \sc H.D. Cornean, G. Nenciu}, {\em On the eigenfunction decay
 for two dimensional magnetic Schr\"odinger operators.} Commun. Math.
Phys. {\bf
 192} (1998), 671-685

\bibitem{CN2}  { \sc H.D. Cornean, G. Nenciu}, {\em The Faraday effect revisited : Thermodynamic limit.} J. Funct. Anal. {\bf
 257} (2009), no. 7, 2024-2066

\bibitem{CFKS} {\sc H.L. Cycon, R.G. Froese, W. Kirsch, B. Simon}, {\em Schr\"odinger operators, with Application to Quantum Mechanics and Global Geometry}. New York : Springer 1987

\bibitem{DS}  { \sc N. Dunford and J.T. Schwartz}, {\em Linear Operators, Part II : Spectral Theory, Self Adjoint Operators in Hilbert Space.}
  Pure and Applied Mathematics, Interscience Publishers, New York 1963

\bibitem{DoIwMi}  { \sc S. Doi, A. Iwatsuka,T. Mine}, {\em The uniqueness of the integrated  density of states for the Schr\"odinger operatorswith magnetic field.}
  Math. Z., {\bf 237} (2001), 335-371.
\bibitem {GT} {\sc D. Gilbarg, N. S. Trudinger}, {\em Elliptic Partial
  Differential Equations of Second Order}.
 Classics in Mathematics, Springer-Verlag
 New York, Inc., New York 1998

\bibitem{hartog} {\sc L. H\"ormander,} {\em   An Introduction to Complex Analysis in Several Variables,} 3rd edition.
North-Holland Mathematical Library 7, North-Holland 1990


\bibitem{HeSj} { \sc B. Helffer , J. Sj\"ostrand},
{\em On  diamagnetism and de Haas-van-Alphen effect.}
Ann. I.H.P. Phys. Theor. {\bf 52} (1990), 303-352


\bibitem {Hu} { \sc K. Huang}, {\em Statistical Mechanics.} Edition J. Wiley, 1987

\bibitem {HuSi} { \sc D. Hundertmark, B. Simon}, {\em A diamagnetic inequality for semigroup differences.} J. Reine Angew. Math.  {\bf 571}  (2004), 107-130

\bibitem {If} {\sc V. Iftimie}, {\em Uniqueness ans existence of the IDS for Schr\"odinger operators with magnetic field and elecric potential with singular negative part.}
 Publ. Res. Inst. Math. Sci.  {\bf 41}  (2005),  no. 2, 307-327

\bibitem {HG} {\sc A.M.  Hinz, G. Stolz}, {\em Polynomial boundedness of eigensolutions and the spectrum of Schr\"odinger operators}.
 Math. Ann.{ \bf  294}   (1992), 195-211.

\bibitem {K} { \sc T. Kato}, {\em Perturbation Theory for Linear
 Operators.} New York : Springer 1966


\bibitem{RS2} {\sc M. Reed, B. Simon B}, {\em Methods of Modern Mathematical Physics II : Fourier Analysis, Self-adjointness}. New York : Academic Press 1975

\bibitem {RS4} {\sc M. Reed, B. Simon}, {\em Methods of Modern
  Mathematical Physics IV : Analysis of Operators}. New York : Academic Press 1978

\bibitem{Sa} {\sc B. Savoie}, Ph.D thesis, in preparation.

\bibitem{Si} {\sc B. Simon}, {\em Schr\" odinger Semigroup}, Bull. Amer. Math. Soc. {\bf 7} (1982), 447-526

\bibitem{Z} {\sc V. Zagrebnov,} {\em Topics in the Theory of Gibbs Semigroup.}
Leuven Notes in Mathematical and Theoretical Physics. Leuven University Press, 2003.
\end{thebibliography}
 \end{document}